\documentclass[11pt,reqno]{amsart}
\usepackage{style}
\usepackage{microtype}
\renewcommand{\P}{\mathbb{P}}

\numberwithin{equation}{section}
\makeatletter
\renewcommand{\tagform@}[1]{\maketag@@@{\normalcolor(\ignorespaces#1\unskip\@@italiccorr)}}
\makeatother

\usepackage[
  paper=letterpaper,
  textwidth=450pt,
  includehead,            
  headheight=8pt,
  headsep=14pt,
  textheight=606pt,       
  centering               
]{geometry}

\renewcommand\d{\mathrm{d}}

\title[On the Cardinality of the Smith Set]{On the Cardinality of the Smith Set\\ Under Impartial Culture with Many Alternatives}
\author{Zachary Berchenko}
\address{Saint Paul Academy and Summit School}
\email{zackberchenko@gmail.com}

\author{Kaizhao Liu}
\address{Massachusetts Institute of Technology}
\email{mrzt@mit.edu}

\date{\today}

\begin{document}

\begin{abstract}
For some integer $\ell\geq 2$, consider $m=2\ell-1$ voters, each of which has a preference ranking over $n$ alternatives.
The Smith set is the smallest nonempty set of alternatives each of which defeats every alternative outside the set in a pairwise majority comparison.
Under the standard benchmark \textit{impartial culture} where each voter uniformly and independently chooses a random preference ranking over all alternatives, 
we study the asymptotic probability that the Smith set has cardinality $s$ for fixed $\ell$ and $n\to \infty$.

First, we prove that the probability that the Smith set has constant cardinality $s$ goes to zero at a rate of $\Theta_{\ell,s} (n^{-s(\ell-1)/\ell})$. 
Next, if $\min\{s,n-s\}\to\infty$, we show that this probability decays superpolynomially.
Further, when $s$ and $n-s$ are both $\Theta(n)$, we prove that the probability decays exponentially and determine the exact exponential base. 
Finally, we prove that the probability that the Smith set contains all alternatives approaches one at a rate of $\Theta_\ell(n^{-(\ell-1)/\ell})$.
These theorems resolve the conjectures in \cite{liu2025statistical} under impartial culture.
\end{abstract}

\maketitle

\section{Introduction.}
The aggregation of individual preferences into a collective decision is a foundational problem in social choice theory, with applications spanning economics, political science, and computer science~\cite{shoham2008multiagent,easley2010networks,arrow2012social,sen2017collective}. 
In scenarios where a population votes on a slate of alternatives, the most natural objective is to identify a unique, unambiguously ``best'' option which defeats every other alternative in pairwise majority comparisons, known as the \textit{Condorcet winner}~\cite{gehrlein2006condorcet,xiao2025theoretical}. 
However, a unique Condorcet winner does not always exist.
We identify the alternatives with $[n]=\{1,\ldots,n\}$, so a ranking is a permutation of $[n]$. 
In the classic \textit{Condorcet paradox}~\cite{de2014essai}, three voters have cyclic preferences over the alternatives $1,2,3$: $1\succ2\succ3$, $2\succ3\succ1$, and $3\succ1\succ2$. Then $1$ defeats $2$, $2$ defeats $3$, and $3$ defeats $1$ by pairwise majority, leaving no alternative undefeated.

When majority rule fails to produce a single Condorcet winner, the standard social choice approach relaxes the requirement of a single champion to a \textit{set} of top-tier candidates that collectively dominate the remaining field.
This is formalized through the notion of a \textit{dominating set} in the \textit{majority digraph}.
\begin{definition}
The \textit{majority digraph} is the directed graph on $[n]$ with an arc from $i$ to $k$ if and only if a majority of voters prefer alternative $i$ over alternative $k$.
A nonempty set \(D\subseteq[n]\) is said to be \textit{dominating} if, for every \(i\in D\) and \(k\in[n]\setminus D\), the majority digraph contains the arc \(i\to k\).
\end{definition}

Dominating sets are generally not unique. For example, the entire set of alternatives is always a dominating set.
However, dominating sets are nested, forming a strict chain under set inclusion. 
Rather than selecting an arbitrary or overly conservative dominating set, social choice theory focuses on the minimal element of this chain, universally known as the \textit{Smith set}~\cite{smith1973aggregation,shi2025fundamental}.

\begin{definition}
    The \textit{Smith set} $S$ is the smallest dominating set of alternatives.
\end{definition}
The Smith set is unique and always exists when the number of alternatives is finite~\cite{liu2025statistical}. 
If a Condorcet winner exists, then the Smith set contains only that alternative.
For an odd number of voters, the majority digraph is a tournament\footnote{A tournament is a simple digraph where every pair of distinct vertices is connected by exactly one arc.}, and the Smith set is its unique strongly connected component that dominates all the others~\cite{Brandt_Conitzer_Endriss_Lang_Procaccia_2016}. 
A tournament on two vertices is not strongly connected, so the Smith set has cardinality either one or at least three~\cite{liu2025statistical}. 
By Camion's theorem~\cite{camion1959chemins}, every strongly connected tournament on at least three vertices has a Hamiltonian cycle. 
Thus, for $n\ge3$, the Smith set equals $[n]$ if and only if the majority tournament has a Hamiltonian cycle.
We say that a set of alternatives forms a \textit{majority cycle} if their induced subgraph is Hamiltonian. 
Thus, a set of at least three alternatives is the Smith set if and only if it forms a majority cycle and is dominating.

In general, the size of the Smith set serves as a structural measure of consensus:
a small set identifies a clear top tier, while a large set indicates widespread majority cycles. 
The primary goal of this paper is to study the typical cardinality of the Smith set. 

\subsection{A Probabilistic Voting Model}
To investigate this mathematically, we employ the \textit{impartial culture} model~\cite{garman1968paradox,tsetlin2003impartial,sauermann2022probability}, a foundational benchmark in social choice theory, where each voter independently chooses a preference ranking uniformly at random from all $n!$ possible orderings of $n$ alternatives. Under this random voting model, the macroscopic behavior of preference aggregation depends on the magnitude of the population size $m$ and the number of alternatives $n$.

Traditionally, social choice theory studies the behavior of preference aggregation in a regime where a large, growing population of voters ($m \to \infty$) evaluates a small, fixed number of alternatives ($n$). This is particularly the case in voting and elections, where a large population votes for few candidates (alternatives).
Recent applications, however, necessitate the study of a different regime~\cite{liu2025statistical}. For example, in the alignment of large language models via Reinforcement Learning from Human Feedback~\cite{Ouyang2022}, a fixed panel of human expert labelers (voters) provides pairwise comparisons over a massive pool of potential model outputs (alternatives). This structure motivates the study of the regime where the number of voters $m$ is fixed, while the number of alternatives $n$ tends to infinity~\cite{liu2025statistical}.
This leads to the following central question in the paper:
\begin{center}
    \emph{Let $S_{m, n}$ be the Smith set under an impartial culture with $m$ voters and $n \to \infty$ alternatives. For any integer $3 \le s \le n$, what is the asymptotic probability that $|S_{m, n}| = s$?}
\end{center}

Although the impartial culture model has been a standard benchmark in social choice theory for over half a century, the asymptotic behavior of the Smith set as $n \to \infty$ with fixed $m$ has remained largely untouched until recently. 
Most of prior works focus on the classical regime, where the analysis relies heavily on the Central Limit Theorem: as the number of voters approaches infinity, the pairwise preference margins between candidates converge to multivariate normal distributions~\cite{bellRandomVotingGraph1981,mosselGaussianBoundsNoise2010}.
In particular,  \cite{bellRandomVotingGraph1981} proved in 1981 that, under an impartial culture with many voters, the probability that the Smith set is the entire set of alternatives, i.e. a Hamiltonian cycle exists in the majority digraph, converges to $1$ as the number of alternatives increases.

For the regime considered in this paper, the Gaussian approximations of the traditional literature no longer apply, and analyzing the size of the Smith set becomes a difficult problem of probabilistic combinatorics. 
Prior works focus exclusively on the case where $s=1$, i.e. there exists a Condorcet winner.
A long-standing conjecture by Garman and Kamien \cite{garman1968paradox} posited that the probability of a unique Condorcet winner existing vanishes as $n \to \infty$ based on numerical evidence. 
Later, \cite{may1971some} proved that this probability converges to zero at a rate of $O(n^{-\frac{1}{2}})$ for any fixed $m\geq 3$, but it is unclear whether this bound is tight. 
This was recently resolved in \cite{sauermann2022probability}, where it is proved that for $m = 2\ell-1$ voters, the probability of a unique Condorcet winner decays at a tight polynomial rate of $\Theta_\ell(n^{-\frac{\ell-1}{\ell}})$.
Beyond this singleton case, the asymptotic distribution of the Smith-set cardinality has remained unknown.

\subsection{Main Results}
For $\mathbf w=(w_1,\dots,w_m)$, define the elementary symmetric polynomial
\begin{equation}\label{eq:sigma}
\sigma_\ell(\mathbf w)
=\sum_{\substack{M\subseteq[m]\\ |M|=\ell}}\prod_{j\in M}w_j,
\end{equation}
and the constants
\begin{equation}\label{eq:constants}
\begin{aligned}
c_\ell
&=\int_{[0,\infty)^m}e^{-\sigma_\ell(\mathbf w)}\,\d\mathbf w,\\
\lambda_\ell
&=e^{(3\ell-1)/\ell}\binom{m}{\ell}^{m/\ell}c_\ell.
\end{aligned}
\end{equation}
The finiteness of $c_\ell$ is established in \cite[Lemma~3]{sauermann2022probability}.
Using the constants in~\eqref{eq:constants}, we upper bound the probability that the Smith set has any cardinality $s<n$.
\begin{theorem}\label{thm:upper bound s=o(n)}
For all integers $1\le s\le n/2$,
\[
\max\bigl\{\P(|S_{m,n}|=s),\,\P(|S_{m,n}|=n-s)\bigr\}
\le \lambda_\ell^s
     \left(\frac{s}{n}\right)^{s(\ell-1)/\ell}.
\]
\end{theorem}

With Theorem~\ref{thm:upper bound s=o(n)}, we resolve the following conjecture from \cite{liu2025statistical} on the sharp decay rate of the probability that the Smith set has any fixed cardinality $s$.
\begin{theorem}\label{thm:constant s}
    For any fixed cardinality $s \ge 3$,
\[
\P(|S_{m, n}|=s)=\Theta_{\ell,s} \left(n^{-\frac{s(\ell-1)}{\ell}}\right).
\]
\end{theorem}
Theorem~\ref{thm:upper bound s=o(n)} also gives a superpolynomial upper bound whenever $t=\min\{s,n-s\}\to\infty$ and $t=o(n)$.
Indeed, the logarithm of the right-hand side in Theorem~\ref{thm:upper bound s=o(n)} is
\[
t\ln\lambda_\ell-\frac{\ell-1}{\ell}t\ln\frac nt
=-\left(\frac{\ell-1}{\ell}+o(1)\right)t\ln\frac nt.
\]
Since $t\to\infty$ and $t=o(n)$ imply $t\ln(n/t)/\ln n\to\infty$, this logarithm is $-\omega(\ln n)$.

To treat the remaining case $t=\Theta(n)$, define the lower binomial tail, for $p\in(0,1)$ and integer $0\le k\le m$, by
\begin{equation}\label{eq:binomial-tail}
F_{m,p}(k)=\sum_{i=0}^k\binom{m}{i}p^i(1-p)^{m-i}
\end{equation}
and, for $\alpha\in[0,1]$, define
\begin{equation}\label{eq:zeta}
\zeta_\ell(\alpha)
=\sup_{\substack{p\in(0,1)\\ r\in[\ell-1]_0}}
  F_{m,p}(r)^\alpha\,
  F_{m,1-p}(\ell-1-r)^{1-\alpha}.
\end{equation}
Let $\mu_d$ denote $d$-dimensional Lebesgue measure. 
For $0<\alpha<1$, the function $\zeta_\ell(\alpha)$ equals the supremum of $\mu_m(U)^\alpha\mu_m(V)^{1-\alpha}$ over measurable sets $U,V\subseteq[0,1]^m$ such that, for every $\mathbf u\in U$ and $\mathbf v\in V$, at least $\ell$ coordinates $j\in[m]$ satisfy $u_j<v_j$; see~\cite{HARPER1999285}.
We have the following upper bound.
\begin{theorem}\label{thm:upper bound s=theta(n)}
There are constants $C_\ell,N_\ell>0$ such that, for every $n\ge N_\ell$ and integer $3\le s\le n-1$,
\[
\P(|S_{m, n}|=s) \le \binom{n}{s} \left(\zeta_{\ell} \left( \frac sn \right)+C_\ell n^{-1/(m+1)}\right)^n.
\]
\end{theorem}
Theorem~\ref{thm:upper bound s=theta(n)} provides no significant information when $s$ is small compared to $n$, but shows that $\P(|S_{m, n}|=s)$ decreases exponentially in $n$ when both $s$ and $n-s$ are $\Theta(n)$. Together, Theorems~\ref{thm:upper bound s=o(n)} and~\ref{thm:upper bound s=theta(n)} provide novel upper bounds for all $s<n$.
With them, we resolve the following conjecture from \cite{liu2025statistical} on the probability that Smith set is the entire set:

\begin{theorem}\label{thm:hamiltonian cycle}
 The probability that the majority digraph has a Hamiltonian cycle is
\[
1-\Theta_{\ell} \left( n^{-(\ell-1)/\ell} \right).
\]
\end{theorem}
Theorem~\ref{thm:hamiltonian cycle} further allows us to make our probability asymptotics more precise, as it lets us bound the probability that a set of alternatives form a majority cycle. More specifically, we use it to show that the bound from Theorem~\ref{thm:upper bound s=theta(n)} is tight.
\begin{theorem}\label{thm:lower bound s=theta(n)}
As $n\to\infty$, uniformly over integers $n/\sqrt{\ln n}\le s\le n-1$,
\[
\left(\frac{\P(|S_{m,n}|=s)}{\binom ns}\right)^{1/n}=\zeta_\ell\left(\frac sn\right)+O_\ell\left(n^{-1/(m+1)}\right).
\]
\end{theorem}

The remainder of the paper is organized as follows. Section~\ref{sec:notation} gives the setup and notation. 
In Section~\ref{sec:upper bound s=o(n)}, we use an analytical method to bound $\P(|S_{m, n}|=s)$, yielding Theorem~\ref{thm:upper bound s=o(n)} and Theorem~\ref{thm:constant s}.
In Section~\ref{sec:upper bound s=theta(n)}, we use combinatorial methods to prove Theorem~\ref{thm:upper bound s=theta(n)}, after which we prove Theorem~\ref{thm:lower bound s=theta(n)} assuming that Theorem~\ref{thm:hamiltonian cycle} is true.
Then, we prove Theorem~\ref{thm:hamiltonian cycle} in Section~\ref{sec:hamiltonian cycle} by combining Theorem~\ref{thm:upper bound s=o(n)} and Theorem~\ref{thm:upper bound s=theta(n)}.
Section~\ref{sec:lemmas} proves the technical lemmas.

\subsection{Notation}\label{sec:notation}
For a nonnegative integer $k$, let $[k]=\{1,\ldots,k\}$ with $[0]=\varnothing$, and $[k]_0=\{0,\ldots,k\}$. 
Throughout, $m=2\ell-1\ge3$ is a fixed odd number of voters and $n\to\infty$ is the number of alternatives. 
We identify the alternatives with $[n]$ and write $S_{m,n}\subseteq[n]$ for the random Smith set under impartial culture. 

The random ranking profile is the ordered tuple $\boldsymbol\Pi=(\Pi_1,\ldots,\Pi_m)$ of independent uniform permutations of $[n]$; a fixed profile is denoted by $\boldsymbol\pi$. 
We use uppercase letters for random scores, lowercase letters for their realizations, and bold letters for vectors. 
Calligraphic letters denote probability events, while script letters denote families of sets or parameters. 
In particular, write
\[
\begin{aligned}
\mathcal D_s&=\{[s]\text{ is a dominating set}\},\\
\mathcal C_s&=\{[s]\text{ forms a majority cycle}\}.
\end{aligned}
\]
For $s\ge3$, the event $\{S_{m,n}=[s]\}$ equals $\mathcal D_s\cap\mathcal C_s$. We write $|A|$ for the cardinality of a finite set and $\mu_d(R)$ for the Lebesgue measure of a region $R\subseteq\mathbb R^d$.

We use standard $O$, $\Omega$, and $\Theta$ notation as $n\to\infty$. 
Subscripts indicate the parameters on which the implicit constants may depend; for example, $f(n)=O_\ell(g(n))$ means that $|f(n)|\le C_\ell|g(n)|$ for all $n\ge N_\ell$, for constants $C_\ell,N_\ell$ depending only on $\ell$. 
Unless a result fixes $s$, these constants are uniform in $s$ throughout its stated range. 
Constants denoted by $C_\ell$ may be increased from one estimate to the next.

\section{An Analytic Bound}\label{sec:upper bound s=o(n)}
In this section, we begin by proving Theorem~\ref{thm:upper bound s=o(n)}, which upper bounds the probability that the Smith set has cardinality $s$ when $s$ or $n-s$ is small compared to $n$. Afterwards, we use similar techniques to prove Theorem~\ref{thm:constant s}, which provides the asymptotic probability that the Smith set has any constant cardinality.

\subsection{Proof of Theorem~\ref{thm:upper bound s=o(n)}}
Assume $1\le s\le n/2$. Since the Smith set is unique and all alternatives are symmetric under impartial culture,
\begin{align*}
\P(|S_{m,n}|=s)&{}= \binom ns \P(S_{m,n}=[s])\\
&{}\le \binom ns \P(\mathcal D_s),
\end{align*}
where the inequality uses the fact that $[s]$ being the Smith set implies it is dominating. 
For each voter $j\in[m]$, sample $n$ independent scores $X_{1j},\ldots,X_{nj}\sim\operatorname{Unif}(0,1)$, independently across voters, and rank the alternatives in increasing order of these scores. 
The resulting rankings are independent and uniform, giving impartial culture~\cite{sauermann2022probability,liu2025statistical}. Write
\[
\mathbf X_i=(X_{i1},\ldots,X_{im}),\qquad
\mathbf X_{[s]}=(\mathbf X_1,\ldots,\mathbf X_s).
\]
We use $\mathbf x_i=(x_{i1},\ldots,x_{im})$ and $\mathbf x=(\mathbf x_1,\ldots,\mathbf x_s)\in[0,1]^{ms}$ for their realizations. For each fixed $\mathbf x$, use the conditional law obtained by freezing $\mathbf X_{[s]}=\mathbf x$ and leaving all outside scores independent and uniform on $(0,1)$. This specifies a regular conditional distribution for every $\mathbf x$. To shorten its notation, for any event $\mathcal A$ write
\[
\P(\mathcal A\mid\mathbf x)
:=\P(\mathcal A\mid\mathbf X_{[s]}=\mathbf x).
\]
For each $i\in[n]\setminus[s]$, let $\mathcal A_i$ be the event that every alternative $k\in[s]$ defeats $i$ in the majority comparison.
Since $\mathbf X_{[s]}$ has density $1$ on $[0,1]^{ms}$, conditioning gives
\begin{align*}
\P(\mathcal D_s)=\P\left(\bigcap_{i=s+1}^n \mathcal A_i\right)&{}= \int_{[0,1]^{ms}}
     \P\left(\bigcap_{i=s+1}^n \mathcal A_i\,\middle|\,\mathbf x\right)\d\mathbf x\\
&{}= \int_{[0,1]^{ms}} \P(\mathcal A_n\mid\mathbf x)^{n-s}\d\mathbf x,
\end{align*}
where the second equality uses conditional independence: given $\mathbf X_{[s]}=\mathbf x$, each $\mathcal A_i$ depends only on the outside score vector $\mathbf X_i$, and these vectors remain independent and identically distributed. 
Let $g_s(\mathbf x)$ be the conditional probability that alternative $n$ defeats at least one alternative in $[s]$. Then $\P(\mathcal A_n\mid\mathbf x)=1-g_s(\mathbf x)$, so
\[
\P\left(\bigcap_{i=s+1}^n \mathcal A_i\right)
  = \int_{[0,1]^{ms}} (1-g_s(\mathbf x))^{n-s}\d\mathbf x.
\]
The following lemma provides a bound on $g_s(\mathbf x)$ by the polynomial in \eqref{eq:analytic-polynomial}.

\begin{lemma}\label{lemma:5}
Let
\begin{equation}
P_s(\mathbf x)
  = \sum_{i=1}^s\sum_{\substack{M\subseteq[m]\\|M|=\ell}}
    \prod_{j\in M}x_{ij}.
\label{eq:analytic-polynomial}
\end{equation}
Then, for all $\mathbf x\in[0,1]^{ms}$,
\[
P_s(\mathbf x)\ge g_s(\mathbf x)
  \ge \frac{P_s(\mathbf x)}{s\binom m\ell}.%
\]
\end{lemma}
\begin{proof}
Alternative $n$ defeats some $i\in[s]$ if and only if there is a set $M\subseteq[m]$ with $|M|=\ell$ such that all voters $j\in M$ prefer $n$ to $i$. For each such pair $(i,M)$, let $\mathcal E_{i,M}$ be this event; equivalently, $X_{nj}<X_{ij}$ for every $j\in M$. Then,
\[
g_s(\mathbf x)
  = \P\Bigg(
    \bigcup_{i\in[s]}\bigcup_{\substack{M\subseteq[m]\\|M|=\ell}}\mathcal E_{i,M}
    \,\Bigg|\,\mathbf x\Bigg).
\]
For any fixed $i$ and $M$, we have
\[
\P(\mathcal E_{i,M}\mid\mathbf x)
  = \P\Bigg(\bigcap_{j\in M}\{x_{ij}>X_{nj}\}\,\Bigg|\,\mathbf x\Bigg).
\]
Given $\mathbf X_{[s]}=\mathbf x$, each $X_{ij}$ with $i\in[s]$ equals $x_{ij}$, while the scores $X_{nj}$ remain independent across voters. Hence the events $\{X_{nj}<x_{ij}\}$ are conditionally independent, and
\[
\P(\mathcal E_{i,M}\mid\mathbf x)
  = \prod_{j\in M}\P(x_{ij}>X_{nj}\mid\mathbf x)
  = \prod_{j\in M}x_{ij}.
\]
Therefore, for all sets $M\subseteq[m]$ with $|M|=\ell$ and $i\in[s]$,
\begin{equation*}
g_s(\mathbf x)\ge \P(\mathcal E_{i,M}\mid\mathbf x)=\prod_{j\in M}x_{ij}.
\end{equation*}
Summing these inequalities over all $s\binom m\ell$ pairs $(M,i)$ gives us
\[
s\binom m\ell g_s(\mathbf x)\ge P_s(\mathbf x),
\qquad
g_s(\mathbf x)\ge\frac{P_s(\mathbf x)}{s\binom m\ell}.
\]
Similarly,
\[
g_s(\mathbf x)
  \le \sum_{i=1}^s\sum_{\substack{M\subseteq[m]\\|M|=\ell}}
      \P(\mathcal E_{i,M}\mid\mathbf x)
  = P_s(\mathbf x).
\]
\end{proof}
Since $e^{-x}\ge 1-x$ for all $x\in[0,1]$,
\begin{align}
\int_{[0,1]^{ms}}(1-g_s(\mathbf x))^{n-s}\d\mathbf x&{}\le \int_{[0,1]^{ms}}e^{-(n-s)g_s(\mathbf x)}\d\mathbf x\notag\\
&{}\le \int_{[0,1]^{ms}}
    \exp\left(-\frac{n-s}{s\binom m\ell}P_s(\mathbf x)\right)\d\mathbf x\notag\\
&{}\le \left(\frac{n-s}{s\binom m\ell}\right)^{-ms/\ell}
    \int_{[0,\infty)^{ms}}e^{-P_s(\mathbf z)}\d\mathbf z,
\label{eq:analytic-integral-upper}
\end{align}
where we used the substitution $\mathbf z=\bigl((n-s)/(s\binom m\ell)\bigr)^{1/\ell}\mathbf x$ and homogeneity of $P_s$ of degree $\ell$. 
The Jacobian gives the prefactor. 
The transformed domain is a finite cube in the nonnegative orthant; enlarging it to the whole orthant preserves the upper bound because the integrand is nonnegative. 
The final integral can be factorized by the following lemma.
\begin{lemma}\label{lemma:6}
We have
\[
\int_{[0,\infty)^{ms}}e^{-P_s(\mathbf z)}\d\mathbf z=c_\ell^s,
\]
where $c_\ell$ is the finite constant defined in \eqref{eq:constants} and depends only on $\ell$.
\end{lemma}
The proof of Lemma~\ref{lemma:6} involves standard integral simplifications and can be found in Section~\ref{sec:lemmas-1}. Substituting its conclusion into \eqref{eq:analytic-integral-upper} produces
\[
\binom ns\P(\mathcal D_s)
  \le \binom ns\left(\frac{n-s}{s\binom m\ell}\right)^{-ms/\ell}c_\ell^s.
\]
Now, we bound the term involving $n-s$:
\begin{equation*}
(n-s)^{-ms/\ell}{}\le (n/2)^{-ms/\ell}=n^{-ms/\ell}2^{ms/\ell}{}\le n^{-ms/\ell}e^{ms/\ell}.
\end{equation*}
Next, we address the binomial and factorial terms. Applying the standard bound $\binom ns\le n^s/s!$ and the factorial lower bound $s!\ge\sqrt{2\pi s}(s/e)^s\ge(s/e)^s$, we obtain
\[
\binom ns s^{ms/\ell}\le n^s e^s s^{s(m/\ell-1)}.
\]
Therefore, using the definition of $\lambda_\ell$ in \eqref{eq:constants},
\begin{align}
\P(|S_{m,n}|=s)&{}\le\binom ns\P(\mathcal D_s)\notag\\*
&{}\le \left(\frac sn\right)^{s(\ell-1)/\ell}
    \left[c_\ell e^{(3\ell-1)/\ell}\binom m\ell^{m/\ell}\right]^s{}= \lambda_\ell^s\left(\frac sn\right)^{s(\ell-1)/\ell}.
\label{eq:analytic-domination-bound}
\end{align}
To extend \eqref{eq:analytic-domination-bound} to $\P(|S_{m,n}|=n-s)$, consider the following lemma.
\begin{lemma}\label{lem:s to n-s}
For all $n\in\mathbb N$ and all $1\le s\le n-1$,
\[
\P(\mathcal D_s)
  =\P(\mathcal D_{n-s}).
\]
\end{lemma}
\begin{proof}
Let $\boldsymbol\pi$ be an $m$-tuple of rankings on $[n]$, and let $\boldsymbol\pi^{\mathrm{rev}}$ be obtained by reversing each ranking. These profiles have equal probability under impartial culture. Furthermore, if alternative $i$ defeats $k$ under $\boldsymbol\pi$, then $k$ defeats $i$ under $\boldsymbol\pi^{\mathrm{rev}}$. Thus, $[s]$ dominates its complement under $\boldsymbol\pi$ if and only if $[n]\setminus[s]$ dominates $[s]$ under $\boldsymbol\pi^{\mathrm{rev}}$. Hence,%
\[
\P(\mathcal D_s)
  =\P([n]\setminus[s]\text{ is a dominating set}).
\]
Furthermore, $[n]\setminus[s]$ and $[n-s]$ have the same cardinality, so they have the same probability of being a dominating set by symmetry, and the conclusion follows.
\end{proof}
This lemma transfers the same bound to Smith sets of size $n-s$, obtaining the following corollary.
\begin{corollary}\label{cor:Bounding n-s}
For all $1\le s\le n/2$,
\[
\P(|S_{m,n}|=n-s)
  \le\lambda_\ell^s\left(\frac sn\right)^{s(\ell-1)/\ell}.
\]
\end{corollary}
The proof is provided in Section~\ref{sec:lemmas-1}, which applies Lemma~\ref{lem:s to n-s} to the bound in \eqref{eq:analytic-domination-bound}.
This completes the proof of Theorem~\ref{thm:upper bound s=o(n)}.

\subsection{Proof of Theorem~\ref{thm:constant s}}
First, note that when $s$ is constant, Theorem~\ref{thm:upper bound s=o(n)} gives
\[
\P(|S_{m,n}|=s)
  \le\lambda_\ell^s\left(\frac sn\right)^{s(\ell-1)/\ell}
  =O_{\ell,s}\left(n^{-s(\ell-1)/\ell}\right).
\]
Thus, it suffices to show that $\P(|S_{m,n}|=s)=\Omega_{\ell,s}\bigl(n^{-s(\ell-1)/\ell}\bigr)$. 
Fix $s\ge3$ and take $n\ge2s$. Recall that
\begin{align*}
\P(|S_{m,n}|=s)&{}=\binom ns\P(S_{m,n}=[s])\\
&{}=\binom ns\P\left(\mathcal C_s\cap\mathcal D_s\right).
\end{align*}
Let $\boldsymbol\pi^\star$ be an $m$-tuple of rankings on $[s]$ whose majority tournament has a Hamiltonian cycle. 
Such a profile exists for every $s\ge3$: partition $[s]$ into three nonempty disjoint sets $A,B,C$, and have the first three voters rank
\[
A\succ B\succ C,\qquad B\succ C\succ A,\qquad C\succ A\succ B,
\]
using the same fixed ordering within each block for all three voters. 
The majority arcs follow these orders within blocks and run from $A$ to $B$, from $B$ to $C$, and from $C$ to $A$. 
Concatenating the three block orders therefore gives a Hamiltonian cycle. 
If $m>3$, add pairs of voters with reversed rankings; each pair leaves every majority margin unchanged.

For this fixed profile, define the occurrence event
\[
\mathcal E_{\pi^\star}
=\{\boldsymbol\Pi|_{[s]}=\boldsymbol\pi^\star\}.
\]
Uniformity and independence give $\P(\mathcal E_{\pi^\star})=(s!)^{-m}$.
Since $\mathcal E_{\pi^\star}\subseteq\mathcal C_s$,
\begin{align}
\P(|S_{m,n}|=s)&{}\ge\binom ns\P\left(\mathcal C_s\cap\bigcap_{i=s+1}^n \mathcal A_i\,\middle|\,\mathcal E_{\pi^\star}\right)\P(\mathcal E_{\pi^\star})\notag\\
&{}=\frac{\binom ns}{(s!)^m}
    \P\left(\bigcap_{i=s+1}^n \mathcal A_i\,\middle|\,\mathcal E_{\pi^\star}\right).
\label{eq:analytic-profile-bound}
\end{align}
Write $\pi_j^\star(1),\ldots,\pi_j^\star(s)$ for voter $j$'s alternatives in decreasing preference order. 
Define the realization region by
\[
R_{\pi^\star}
=\left\{\mathbf x\in(0,1)^{ms}:
 x_{\pi_j^\star(1),j}<\cdots<x_{\pi_j^\star(s),j}
 \text{ for all }j\in[m]\right\}.
\]
Then the event $\mathcal E_{\pi^\star}$ agrees almost surely with $\{\mathbf X_{[s]}\in R_{\pi^\star}\}$. 
The region is the Cartesian product of $m$ ordered $s$-simplices, each of volume $1/s!$, so $\mu_{ms}(R_{\pi^\star})=(s!)^{-m}$. 
Consequently, the conditional density of $\mathbf X_{[s]}$ given $\mathcal E_{\pi^\star}$ is $(s!)^m$ on this region and zero elsewhere, yielding
\begin{equation}
\P\left(\bigcap_{i=s+1}^n \mathcal A_i\,\middle|\,\mathcal E_{\pi^\star}\right)
  =(s!)^m\int_{R_{\pi^\star}}
    \P\left(\bigcap_{i=s+1}^n \mathcal A_i\,\middle|\,\mathbf x\right)\d\mathbf x.
\label{eq:analytic-profile-density}
\end{equation}
Substituting \eqref{eq:analytic-profile-density} into \eqref{eq:analytic-profile-bound} gives
\begin{equation}
\P(|S_{m,n}|=s)
  \ge\binom ns\int_{R_{\pi^\star}}
    \P\left(\bigcap_{i=s+1}^n \mathcal A_i\,\middle|\,\mathbf x\right)\d\mathbf x.
\label{eq:analytic-smith-integral}
\end{equation}
The events $\mathcal A_i$ are independent and identically distributed given $\mathbf X_{[s]}=\mathbf x$, and $\P(\mathcal A_n\mid\mathbf x)=1-g_s(\mathbf x)$. 
Thus,
\[
\int_{R_{\pi^\star}}
  \P\left(\bigcap_{i=s+1}^n \mathcal A_i\,\middle|\,\mathbf x\right)\d\mathbf x
  =\int_{R_{\pi^\star}}(1-g_s(\mathbf x))^{n-s}\d\mathbf x.
\]
Put $t=(2s\binom m\ell)^{-1/\ell}\in(0,1)$ and let $R'_{\pi^\star}=tR_{\pi^\star}$. 
Positive dilation preserves all the strict ordering inequalities, so $R'_{\pi^\star}=R_{\pi^\star}\cap(0,t)^{ms}\subseteq R_{\pi^\star}$ and
\[
\mu_{ms}(R'_{\pi^\star})=\frac{1}{(s!)^m}\left(2s\binom m\ell\right)^{-ms/\ell}.
\]
Furthermore, the polynomial $P_s$ defined in \eqref{eq:analytic-polynomial} satisfies, for all $\mathbf x\in R'_{\pi^\star}$,
\[
P_s(\mathbf x)
  \le s\binom m\ell\left[\left(2s\binom m\ell\right)^{-1/\ell}\right]^\ell
  =\frac12.
\]
Therefore, since $1-x\ge e^{-2x}$ for all $x\in[0,1/2]$,
\begin{align}
\int_{R_{\pi^\star}}(1-g_s(\mathbf x))^{n-s}\d\mathbf x&{}\ge\int_{R'_{\pi^\star}}(1-g_s(\mathbf x))^{n-s}\d\mathbf x\notag\\
&{}\ge\int_{R'_{\pi^\star}}(1-P_s(\mathbf x))^{n-s}\d\mathbf x\notag\\
&{}\ge\int_{R'_{\pi^\star}}e^{-2(n-s)P_s(\mathbf x)}\d\mathbf x.
\label{eq:analytic-restricted-integral}
\end{align}
Make the substitution $\mathbf z=(n-s)^{1/\ell}\mathbf x$. 
Homogeneity gives $P_s(\mathbf z)=(n-s)P_s(\mathbf x)$ and the Jacobian is $(n-s)^{-ms/\ell}$.
Moreover, $R'_{\pi^\star}\subseteq(n-s)^{1/\ell}R'_{\pi^\star}$: for every $\mathbf z\in R'_{\pi^\star}$, the vector $(n-s)^{-1/\ell}\mathbf z$ satisfies the same strict ordering inequalities and remains in $(0,t)^{ms}$. 
Thus,
\begin{align}
\int_{R'_{\pi^\star}}e^{-2(n-s)P_s(\mathbf x)}\d\mathbf x
&{}=(n-s)^{-ms/\ell}
 \int_{(n-s)^{1/\ell}R'_{\pi^\star}}e^{-2P_s(\mathbf z)}\d\mathbf z\notag\\
&{}\ge(n-s)^{-ms/\ell}
 \int_{R'_{\pi^\star}}e^{-2P_s(\mathbf z)}\d\mathbf z\notag\\
&{}\ge(n-s)^{-ms/\ell}e^{-1}\mu_{ms}(R'_{\pi^\star})\notag\\
&{}=\frac{(n-s)^{-ms/\ell}e^{-1}}{(s!)^m}
 \left(2s\binom m\ell\right)^{-ms/\ell}.
\label{eq:analytic-scaled-integral}
\end{align}
Since $n\ge2s$, we have $\binom ns\ge(n/2)^s/s!$ and $(n-s)^{-ms/\ell}\ge n^{-ms/\ell}$. 
Applying conditional independence in \eqref{eq:analytic-smith-integral}, followed by \eqref{eq:analytic-restricted-integral} and \eqref{eq:analytic-scaled-integral}, therefore gives
\begin{equation}
\P(|S_{m,n}|=s)
\ge\frac{e^{-1}}{2^s(s!)^{m+1}}
 \left(2s\binom m\ell\right)^{-ms/\ell}
 n^{-s(\ell-1)/\ell}.
\label{eq:analytic-fixed-lower}
\end{equation}
The coefficient is strictly positive and depends only on $\ell$ and the fixed $s$, proving the required lower bound.
This completes the proof of Theorem~\ref{thm:constant s}.

\section{A Combinatorial Bound}\label{sec:upper bound s=theta(n)}
In this section, we prove Theorem~\ref{thm:upper bound s=theta(n)}, which provides an upper bound on the probability that the Smith set has cardinality $s=\Theta(n)$. 
We will then prove Theorem~\ref{thm:lower bound s=theta(n)} assuming Theorem~\ref{thm:hamiltonian cycle}, which provides a tight asymptotic bound on this probability. 
This is possible because, while Theorem~\ref{thm:hamiltonian cycle} is proved later in the paper, its proof relies solely on Theorems~\ref{thm:upper bound s=o(n)} and~\ref{thm:upper bound s=theta(n)}.

\subsection{Proof of Theorem~\ref{thm:upper bound s=theta(n)}}
Fix $3\le s\le n-1$ and let $\alpha=s/n\in(0,1)$. 
Set
\[
L=\left\lfloor n^{1/(m+1)}\right\rfloor.
\]
Partition $[0,1)^m$ into $L^m$ half-open hypercubes of side length $1/L$, and for $\mathbf x=(x_1,\dots,x_m)\in[0,1)^m$ let
\begin{equation}\label{eq:comb-grid-map}
    \kappa_L(\mathbf x)
    =\bigl(\lfloor Lx_1\rfloor,\lfloor Lx_2\rfloor,
      \dots,\lfloor Lx_m\rfloor\bigr)
\end{equation}
be its hypercube index. 
Call a pair of deterministic sets $U,V\subseteq[L-1]_0^m$ admissible if, for every $\mathbf u\in U$ and $\mathbf v\in V$, at least $\ell$ coordinates $j\in[m]$ satisfy $u_j\le v_j$. 
The following lemma, with proof provided in Section~\ref{sec:lemmas-2}, controls the sizes of all such pairs.
\begin{lemma}\label{lem:upperbound_set_size}
Let $L\ge1$ be an integer and $0<\alpha<1$. Suppose $U,V\subseteq[L-1]_0^m$ are sets such that, for every $\mathbf u\in U$ and $\mathbf v\in V$, at least $\ell$ coordinates $j\in[m]$ satisfy $u_j\le v_j$. Then
\[
    |U|^\alpha|V|^{1-\alpha}
    \le (L+1)^m\zeta_\ell(\alpha).
\]
\end{lemma}
Let $\mathscr W$ be the family of all admissible pairs.
For each fixed pair $(U,V)\in\mathscr W$, let $\mathcal B(U,V)$ be the event that $\kappa_L(\mathbf X_i)\in U$ for all $i\in[s]$ and $\kappa_L(\mathbf X_k)\in V$ for all $k\in[n]\setminus[s]$. 
Since every cell has probability $L^{-m}$, independence and Lemma~\ref{lem:upperbound_set_size} give
\begin{equation}\label{eq:comb-cell-probability}
\begin{aligned}
    \P\bigl(\mathcal B(U,V)\bigr)
    &=\frac{|U|^s|V|^{n-s}}{L^{mn}}
      =\left(\frac{|U|^\alpha|V|^{1-\alpha}}{L^m}\right)^n\\
    &\le\left(1+\frac1L\right)^{mn}\zeta_\ell(\alpha)^n.
\end{aligned}
\end{equation}
On $\mathcal D_s$, the sets of hypercube indices occupied by the alternatives in $[s]$ and $[n]\setminus[s]$ form an admissible pair: each comparison $X_{ij}<X_{kj}$ implies $\lfloor LX_{ij}\rfloor\le\lfloor LX_{kj}\rfloor$. 
This shows that $\mathcal D_s\subseteq\bigcup_{(U,V)\in\mathscr W}\mathcal B(U,V)$. 
Moreover, $|\mathscr W|\le2^{2L^m}$, because each of $U$ and $V$ is a subset of a grid of size $L^m$. 
Hence the union bound and~\eqref{eq:comb-cell-probability} give
\begin{align}
\P(\mathcal D_s)&{}\le\sum_{(U,V)\in\mathscr W}\P\bigl(\mathcal B(U,V)\bigr)\notag\\
&{}\le2^{2L^m}\left(1+\frac1L\right)^{mn}\zeta_\ell(\alpha)^n.
\label{eq:comb-union-bound}
\end{align}
For $n\ge2^{m+1}$, we have $n^{1/(m+1)}\ge2$ and hence
\[
\frac12 n^{1/(m+1)}\le L\le n^{1/(m+1)}.
\]
Using $\ln(1+t)\le t$ for $t\ge0$, we obtain the uniform bound
\begin{equation}\label{eq:comb-grid-cost}
\begin{aligned}
    2^{2L^m}\left(1+\frac1L\right)^{mn}
    &\le\exp\left(L^m\ln4+\frac{mn}{L}\right)\\
    &\le\exp\left((\ln4+2m)n^{m/(m+1)}\right).
\end{aligned}
\end{equation}
By~\eqref{eq:comb-union-bound} and~\eqref{eq:comb-grid-cost}, the combined cost is $\exp\bigl(O_\ell(n^{m/(m+1)})\bigr)$. Since $\zeta_\ell(\alpha)\le1$ and $e^t\le1+2t$ for $0\le t\le1$, for all sufficiently large $n$ depending only on $\ell$,
\begin{align}
\P(\mathcal D_s)&{}\le\left[\zeta_\ell(\alpha)
    \exp\left((\ln4+2m)n^{-1/(m+1)}\right)\right]^n\notag\\
&{}\le\left[\zeta_\ell(\alpha)
    +2(\ln4+2m)n^{-1/(m+1)}\right]^n.
\label{eq:comb-domination-bound}
\end{align}
Finally, $\P(S_{m,n}=[s])\leq \P(\mathcal D_s)$, symmetry among the alternatives, and~\eqref{eq:comb-domination-bound} give
\begin{align*}
\P(|S_{m,n}|=s)&{}=\binom ns\P(S_{m,n}=[s])\\
&{}\le\binom ns\left[\zeta_\ell(\alpha)
    +2(\ln4+2m)n^{-1/(m+1)}\right]^n.
\end{align*}
This completes the proof of Theorem~\ref{thm:upper bound s=theta(n)}.

\subsection{Proof of Theorem~\ref{thm:lower bound s=theta(n)}}
We use the following elementary bounds, proved in Section~\ref{sec:lemmas-2}.
\begin{lemma}\label{lem:bounding zeta}
For all $\alpha\in[0,1]$,
\[
4^{-m}\le\zeta_\ell(\alpha)\le1.
\]
\end{lemma}

Fix $n/\sqrt{\ln n}\le s\le n-1$, and write $\alpha=s/n$ and $\varepsilon_n=n^{-1/(m+1)}$. 
Take $n$ sufficiently large, depending only on $\ell$, such that $\varepsilon_n\le4^{-m}/2$. 
We first place the alternatives in two score regions that ensure $[s]$ is dominating. 
Conditional on these placements, we obtain a majority cycle on $[s]$ by covering its score region with rectangles whose induced tournaments share a common alternative.

For $p\in(0,1)$ and $r\in[\ell-1]_0$, define the deterministic regions
\[
\begin{aligned}
U
&{}=\{\mathbf u\in[0,1)^m:|\{j\in[m]:u_j<1-p\}|\ge m-r\},\\
V
&{}=\{\mathbf v\in[0,1)^m:|\{j\in[m]:v_j>1-p\}|\ge\ell+r\}.
\end{aligned}
\]
Every $\mathbf u\in U$ and $\mathbf v\in V$ satisfy $u_j<v_j$ in at least $(m-r)+(\ell+r)-m=\ell$ coordinates. 
Their volumes are binomial tails, so the definition~\eqref{eq:zeta} allows us to choose $p,r$ such that
\begin{equation}\label{eq:comb-near-optimum}
\begin{aligned}
\mu_m(U)^\alpha\mu_m(V)^{1-\alpha}
&{}=F_{m,p}(r)^\alpha F_{m,1-p}(\ell-1-r)^{1-\alpha}\\
&{}\ge\zeta_\ell(\alpha)-\varepsilon_n.
\end{aligned}
\end{equation}
Let $\mathcal T_s$ be the event that $\mathbf X_i\in U$ for every $i\in[s]$ and $\mathbf X_i\in V$ for every $i\in[n]\setminus[s]$. 
Then $\mathcal T_s\subseteq\mathcal D_s$, and independence gives
\begin{equation}\label{eq:comb-allocation-probability}
\P(\mathcal T_s)=\mu_m(U)^s\mu_m(V)^{n-s}
\ge\bigl(\zeta_\ell(\alpha)-\varepsilon_n\bigr)^n>0.
\end{equation}
Conditional on $\mathcal T_s$, the vectors $(\mathbf X_i)_{i\in[s]}$ are independent and uniform on $U$. 
We will show that they form a majority cycle with conditional probability at least $1/2$.

For each $M\subseteq[m]$ with $|M|=m-r$, define the rectangle
\[
U_M=\{\mathbf u\in[0,1)^m:u_j<1-p\text{ for every }j\in M\}.
\]
These rectangles cover $U$ and all contain the cube $[0,1-p)^m$. 
Let $J_M$ be the random set of alternatives in $[s]$ whose score vectors lie in $U_M$:
\[
J_M=\{i\in[s]:\mathbf X_i\in U_M\}.
\]
Let $\mathcal C_M$ be the event that the majority digraph induced by $J_M$ has a Hamiltonian cycle, and let $\mathcal I_s$ be the event that for some $i_0\in[s]$, $\mathbf X_{i_0}$, lies in $[0,1-p)^m$. On $\mathcal T_s$, the sets $J_M$ cover $[s]$. 
If $\mathcal I_s$ and every $\mathcal C_M$ occur, all these cycles contain $i_0$ and every alternative can therefore reach $i_0$ and be reached from $i_0$ along directed paths. 
The tournament on $[s]$ is strongly connected, so Camion's theorem gives $\mathcal C_s$. 
Thus the union bound yields
\begin{equation}\label{eq:comb-cycle-failure}
\P(\mathcal C_s^c\mid\mathcal T_s)
\le\P(\mathcal I_s^c\mid\mathcal T_s)
+\sum_{\substack{M\subseteq[m]\\|M|=m-r}}
\P(\mathcal C_M^c\mid\mathcal T_s).
\end{equation}

Conditional on $\mathcal T_s$,
\[
|J_M|\sim\operatorname{Bin}\!\left(s,\frac{\mu_m(U_M)}{\mu_m(U)}\right).
\]
We first bound the failure probability for each rectangle. 
There are at most $2^m$ rectangles, all of equal volume, and they cover $U$. 
Hence $\mu_m(U_M)/\mu_m(U)\ge2^{-m}$. 
Hoeffding's inequality gives
\begin{equation}\label{eq:comb-rectangle-size}
\P\left(|J_M|<\frac{s}{2^{m+1}}\,\middle|\,\mathcal T_s\right)
\le\exp\left(-\frac{s}{2^{2m+1}}\right).
\end{equation}
For any fixed set $I\subseteq[s]$ with $\P(\mathcal T_s\cap\{J_M=I\})>0$, conditioning on $\mathcal T_s$ and $J_M=I$ makes the vectors $(\mathbf X_i)_{i\in I}$ independent and uniform on $U_M$. 
Rescaling each coordinate interval of this rectangle to $[0,1)$ preserves all rankings, so the induced voting profile on $I$ has the impartial-culture distribution. 
For $|I|\ge s/2^{m+1}$, Theorem~\ref{thm:hamiltonian cycle} therefore gives failure probability at most $C_\ell|I|^{-(\ell-1)/\ell}$. 
Averaging over $I$ and using~\eqref{eq:comb-rectangle-size}, we obtain,
\begin{equation}\label{eq:comb-rectangle-failure}
\begin{aligned}
\P(\mathcal C_M^c\mid\mathcal T_s)
&{}\le\exp\left(-\frac{s}{2^{2m+1}}\right)
+C_\ell\left(\frac{s}{2^{m+1}}\right)^{-(\ell-1)/\ell}\\
&{}\le C_\ell s^{-(\ell-1)/\ell},
\end{aligned}
\end{equation}
where the $C_\ell$ in the second line is increased by a constant factor. 

It remains to find a common alternative. Since $\mu_m(V)\le1$, \eqref{eq:comb-near-optimum} and Lemma~\ref{lem:bounding zeta} imply
\begin{equation}\label{eq:comb-region-mass}
\mu_m(U)
\ge\bigl(\zeta_\ell(\alpha)-\varepsilon_n\bigr)^{1/\alpha}
\ge\left(\frac{4^{-m}}2\right)^{\sqrt{\ln n}}.
\end{equation}
Here we used $1/\alpha=n/s\le\sqrt{\ln n}$. 
Also, $m-r\ge\ell$ gives
\[
\mu_m(U)=\sum_{i=0}^r\binom mi p^i(1-p)^{m-i}
\le2^m(1-p)^\ell.
\]
As $m<2\ell$, it follows that
\[
\frac{(1-p)^m}{\mu_m(U)}
\ge\frac{(1-p)^{2\ell}}{\mu_m(U)}
\ge4^{-m}\mu_m(U).
\]
Combining this with~\eqref{eq:comb-region-mass} and $s\ge n/\sqrt{\ln n}$ gives, for sufficiently large $n$ depending only on $\ell$,
\[
\frac{s(1-p)^m}{\mu_m(U)}
\ge\frac{n}{4^m\sqrt{\ln n}}
\left(\frac{4^{-m}}2\right)^{\sqrt{\ln n}}
\ge\sqrt n.
\]
The last inequality follows by taking logarithms: the middle expression has logarithm $\ln n-O_\ell(\sqrt{\ln n}+\ln\ln n)$.
Conditional independence now gives
\begin{equation}\label{eq:comb-common-box-error}
\begin{aligned}
\P(\mathcal I_s^c\mid\mathcal T_s)
&{}=\left(1-\frac{(1-p)^m}{\mu_m(U)}\right)^s\\
&{}\le\exp\left(-\frac{s(1-p)^m}{\mu_m(U)}\right)
\le e^{-\sqrt n}.
\end{aligned}
\end{equation}

Substituting~\eqref{eq:comb-rectangle-failure} and~\eqref{eq:comb-common-box-error} into~\eqref{eq:comb-cycle-failure}, and using the bound $2^m$ on the number of rectangles, yields
\[
\P(\mathcal C_s^c\mid\mathcal T_s)
\le e^{-\sqrt n}+C_\ell s^{-(\ell-1)/\ell}
\le\frac12
\]
for all sufficiently large $n$, uniformly in the stated range of $s$. On $\mathcal T_s\cap\mathcal C_s$, the set $[s]$ is the Smith set. Exchangeability and~\eqref{eq:comb-allocation-probability} therefore give
\[
\P(|S_{m,n}|=s)
\ge\binom ns\P(\mathcal T_s)\P(\mathcal C_s\mid\mathcal T_s)
\ge\frac12\binom ns\bigl(\zeta_\ell(\alpha)-\varepsilon_n\bigr)^n.
\]
Taking $n$th roots and using $2^{-1/n}\ge1-(\ln2)/n$ and $0<\zeta_\ell(\alpha)-\varepsilon_n\le1$, we conclude that
\[
\begin{aligned}
\left(\frac{\P(|S_{m,n}|=s)}{\binom ns}\right)^{1/n}
&{}\ge\bigl(\zeta_\ell(\alpha)-\varepsilon_n\bigr)2^{-1/n}\\
&{}\ge\zeta_\ell(\alpha)-\varepsilon_n-\frac{\ln2}{n}.
\end{aligned}
\]
Recall $\varepsilon_n=n^{-1/(m+1)}$. Since $n^{-1}=o(\varepsilon_n)$, this is the required lower bound. 
The matching upper bound follows from Theorem~\ref{thm:upper bound s=theta(n)}.

\section{The Probability of a Hamiltonian Cycle}\label{sec:hamiltonian cycle}
We prove Theorem~\ref{thm:hamiltonian cycle} using the analytic and combinatorial upper bounds, Theorems~\ref{thm:upper bound s=o(n)} and~\ref{thm:upper bound s=theta(n)} respectively. 
Throughout this section, write
\begin{equation}\label{eq:hamiltonian-cutoffs}
\beta=\frac{\ell-1}{\ell},\qquad
\delta=(2\lambda_\ell)^{-1/\beta}.
\end{equation}
The constant $\delta$ satisfies the following inequality, which follows directly from the definition of $\lambda_\ell$ in~\eqref{eq:constants}.
\begin{lemma}\label{lem:smallConstant}
\[
0<\delta=(2\lambda_\ell)^{-\ell/(\ell-1)}<\frac12.
\]
\end{lemma}

The following lemma makes Theorem~\ref{thm:upper bound s=o(n)} uniform in the small-cardinality regime.
\begin{lemma}\label{lem:simplified upper bound s=o(n)}
There are constants $K_\ell,N_\ell>0$ such that, for every integer $n\ge N_\ell$ and every integer $3\le s\le\delta n$,
\[
\max\bigl\{\P(|S_{m,n}|=s),\P(|S_{m,n}|=n-s)\bigr\}
\le K_\ell n^{-m/\ell}.
\]
\end{lemma}
The remaining cardinalities are controlled by the next lemma. Define
\begin{equation}\label{eq:central-constants}
\epsilon_\ell=\min\left\{\frac1{\ell\binom m\ell},
  \frac{\delta(1-\delta)^{(1-\delta)/\delta}}2\right\},
\qquad
\rho_\ell=\max\{1-\epsilon_\ell^m,2^{-\delta}\}.
\end{equation}
Here $0<\epsilon_\ell<1/4$ and $0<\rho_\ell<1$, because $0<\delta<1/2$.
\begin{lemma}\label{lem:exponential_bound}
For every $\alpha\in[\delta,1-\delta]$,
\[
\zeta_\ell(\alpha)
\le\rho_\ell\alpha^\alpha(1-\alpha)^{1-\alpha}.
\]
Consequently, for every integer $n\ge1$, every integer $s$ with $\delta n\le s\le(1-\delta)n$, and every real $t\ge0$,
\[
\binom ns\bigl(\zeta_\ell(s/n)+t\bigr)^n
\le(\rho_\ell+2t)^n.
\]
\end{lemma}

The proofs of the above lemmas are provided in Section~\ref{sec:lemmas-3}. 
Theorem~\ref{thm:upper bound s=theta(n)} states that there exist $C_\ell>0$ and $N_\ell>0$ such that, for all $n\geq N_\ell$ and integers $3\le s\le n-1$,
\begin{equation}\label{eq:hamiltonian-combinatorial-upper}
\P(|S_{m,n}|=s)
\le\binom ns\bigl(\zeta_\ell(s/n)+C_\ell n^{-1/(m+1)}\bigr)^n.
\end{equation}
Increase $N_\ell$ so that $C_\ell n^{-1/(m+1)}\le(1-\rho_\ell)/4$. 
For every integer $s$ with $\delta n\le s\le(1-\delta)n$, applying Lemma~\ref{lem:exponential_bound} to~\eqref{eq:hamiltonian-combinatorial-upper} gives
\begin{equation}\label{eq:hamiltonian-central-range}
\P(|S_{m,n}|=s)
\le\left(\frac{1+\rho_\ell}{2}\right)^n
\le n^{-m/\ell},
\end{equation}
where the last inequality holds for all sufficiently large $n$ depending only on $\ell$. 
Every integer $3\le s\le n-3$ either satisfies $s\le\delta n$, satisfies $n-s\le\delta n$, or lies in the central interval. 
Thus, combining Lemma~\ref{lem:simplified upper bound s=o(n)} with~\eqref{eq:hamiltonian-central-range} and increasing $K_\ell$ and $N_\ell$ if necessary yields
\begin{equation}\label{eq:hamiltonian-interior-cardinalities}
\P(|S_{m,n}|=s)\le K_\ell n^{-m/\ell}
\qquad(3\le s\le n-3,\ n\ge N_\ell).
\end{equation}
Summing~\eqref{eq:hamiltonian-interior-cardinalities} for $n\ge\max\{N_\ell,6\}$ gives
\begin{equation}\label{eq:hamiltonian-interior-sum}
\sum_{s=3}^{n-3}\P(|S_{m,n}|=s)
\le(n-5)K_\ell n^{-m/\ell}
\le K_\ell n^{-\beta}.
\end{equation}
Recall that the Smith set cannot have cardinality two. 
Applying Theorem~\ref{thm:upper bound s=o(n)} with the smaller cardinalities $1$ and $2$ gives
\begin{equation}\label{eq:hamiltonian-boundary-cardinalities}
\begin{aligned}
\P(|S_{m,n}|=1),\ \P(|S_{m,n}|=n-1)
  &\le\lambda_\ell n^{-\beta},\\
\P(|S_{m,n}|=n-2)
  &\le 2^{2\beta}\lambda_\ell^2 n^{-2\beta}.
\end{aligned}
\end{equation}
For $n\ge3$, Hamiltonicity is equivalent to $|S_{m,n}|=n$. 
Combining~\eqref{eq:hamiltonian-interior-sum} and~\eqref{eq:hamiltonian-boundary-cardinalities}, the failure probability satisfies
\begin{align}
1-\P(|S_{m,n}|=n)&{}=\sum_{s=1}^{n-1}\P(|S_{m,n}|=s)\notag\\
&{}\le\bigl(2\lambda_\ell+2^{2\beta}\lambda_\ell^2+K_\ell\bigr)n^{-\beta}.
\label{eq:hamiltonian-failure-upper}
\end{align}
For the reverse bound, \cite[Theorem~1]{sauermann2022probability} gives the more precise asymptotic
\begin{equation}\label{eq:hamiltonian-condorcet-asymptotic}
\P(|S_{m,n}|=1)
=c_\ell n^{-\beta}
 +O_\ell\left(\frac{(\ln n)^{1/\ell}}n\right),
\end{equation}
where $c_\ell>0$ is the integral defined in~\eqref{eq:constants}. The error term in~\eqref{eq:hamiltonian-condorcet-asymptotic} divided by $n^{-\beta}$ tends to zero, since $(\ln n/n)^{1/\ell}\to0$. Therefore, for all sufficiently large $n$ depending only on $\ell$,
\begin{equation}\label{eq:hamiltonian-failure-lower}
1-\P(|S_{m,n}|=n)
\ge\P(|S_{m,n}|=1)
\ge\frac{c_\ell}{2}n^{-\beta}.
\end{equation}
The bounds~\eqref{eq:hamiltonian-failure-upper} and~\eqref{eq:hamiltonian-failure-lower} prove Theorem~\ref{thm:hamiltonian cycle}.

\section{Proofs of Technical Lemmas and Corollaries}\label{sec:lemmas}

\subsection{Proofs of Lemma~\ref{lemma:6} and Corollary~\ref{cor:Bounding n-s}}\label{sec:lemmas-1}

\begin{proof}[Proof of Lemma~\ref{lemma:6}]
Write $\mathbf z_i=(z_{i1},\ldots,z_{im})$. Using the definitions of $P_s$ in \eqref{eq:analytic-polynomial} and $\sigma_\ell$ in \eqref{eq:sigma} gives
\[
    P_s(\mathbf z)
    =\sum_{i=1}^s\sum_{\substack{M\subseteq[m]\\|M|=\ell}}
        \prod_{j\in M}z_{ij}
    =\sum_{i=1}^s\sigma_\ell(\mathbf z_i).
\]
Hence,
\begin{equation*}
\int_{[0,\infty)^{ms}} e^{-P_s(\mathbf z)}\,\d\mathbf z{}=\prod_{i=1}^s
        \int_{[0,\infty)^m}e^{-\sigma_\ell(\mathbf z_i)}\,\d\mathbf z_i{}=c_\ell^s.
\end{equation*}
The first equality follows from Tonelli's theorem and the second equality uses the definition of $c_\ell$ in \eqref{eq:constants}. 
\end{proof}

\begin{proof}[Proof of Corollary~\ref{cor:Bounding n-s}]
For $1\le s\le n/2$, \eqref{eq:analytic-domination-bound} bounds:
\[
    \binom ns\P(\mathcal D_s)
    \le\lambda_\ell^s\left(\frac{s}{n}\right)^{s(\ell-1)/\ell}.
\]
Thus, by Lemma~\ref{lem:s to n-s},
\[
    \lambda_\ell^s\left(\frac{s}{n}\right)^{s(\ell-1)/\ell}
    \ge\binom ns\P(\mathcal D_s)
    =\binom ns\P(\mathcal D_{n-s})
    =\binom n{n-s}\P(\mathcal D_{n-s})
    \ge\P(|S_{m,n}|=n-s).
\]
\end{proof}

\subsection{Proofs of Lemmas~\ref{lem:upperbound_set_size} and~\ref{lem:bounding zeta}}\label{sec:lemmas-2}

\begin{proof}[Proof of Lemma~\ref{lem:upperbound_set_size}]
If either set is empty, the claim is immediate. 
Assume both are nonempty.
Set $\widetilde V=V+(1,\ldots,1)$.
Both $U$ and $\widetilde V$ lie in $[L]_0^m$, and $|\widetilde V|=|V|$. 
For every $\mathbf u\in U$ and $\mathbf v\in V$, the hypothesis gives $u_j<v_j+1$ in at least $\ell$ coordinates. 
Thus the Hamming distance, which counts the coordinates in which two vectors differ, is at least $\ell$ between any point of $U$ and any point of $\widetilde V$. 
Let $N_{\ell-1}(U)$ be the set of points in $[L]_0^m$ at Hamming distance at most $\ell-1$ from some point of $U$. Then
\begin{equation}\label{eq:comb-separated-sets}
    \widetilde V\cap N_{\ell-1}(U)=\varnothing,
    \qquad
    |V|\le (L+1)^m-|N_{\ell-1}(U)|.
\end{equation}
To bound this neighborhood, for each integer $r\in[\ell-1]_0$ let $p_r\in(0,1)$ be the unique solution of
\begin{equation}\label{eq:comb-matched-tail}
    F_{m,p_r}(r)=\frac{|U|}{(L+1)^m}.
\end{equation}
This solution exists because $0<|U|<(L+1)^m$ and $F_{m,p}(r)$ is continuous and strictly decreasing in $p$, from $1$ to $0$. 
Harper's isoperimetric theorem~(\cite[Theorem~3]{HARPER1999285}; see also \cite[Lemma~3]{dai2023fundamental}) implies
\begin{equation}\label{eq:comb-hamming-neighborhood}
    \frac{|N_{\ell-1}(U)|}{(L+1)^m}
    \ge\min_{r\in[\ell-1]_0}F_{m,p_r}(r+\ell-1).
\end{equation}
Choose an $r$ attaining the minimum in~\eqref{eq:comb-hamming-neighborhood}. 
Combining this bound with~\eqref{eq:comb-separated-sets} and complementing the binomial tail yields
\begin{align}
\frac{|V|}{(L+1)^m}
&{}\le 1-F_{m,p_r}(r+\ell-1)\notag\\
&{}=F_{m,1-p_r}(\ell-1-r).
\label{eq:comb-complement-bound}
\end{align}
Finally, \eqref{eq:comb-matched-tail}, \eqref{eq:comb-complement-bound}, and the definition~\eqref{eq:zeta} give
\begin{align*}
\frac{|U|^\alpha|V|^{1-\alpha}}{(L+1)^m}
&{}\le F_{m,p_r}(r)^\alpha
    F_{m,1-p_r}(\ell-1-r)^{1-\alpha}\\
&{}\le\zeta_\ell(\alpha),
\end{align*}
which proves the claim.
\end{proof}

\begin{proof}[Proof of Lemma~\ref{lem:bounding zeta}]
For $p\in(0,1)$ and $r\in[\ell-1]_0$, both binomial tails in~\eqref{eq:zeta} lie in $(0,1)$. Thus, for every $\alpha\in[0,1]$, we have $\zeta_\ell(\alpha)\le1$ and
\[
    \zeta_\ell(\alpha)
    \ge\sup_{\substack{p\in(0,1)\\r\in[\ell-1]_0}}
        F_{m,p}(r)\,F_{m,1-p}(\ell-1-r).
\]
Taking $p=\tfrac12$ in~\eqref{eq:binomial-tail}, we have, for every integer $t\in[\ell-1]_0$,
\[
    F_{m,1/2}(t)=2^{-m}\sum_{i=0}^t\binom mi\ge2^{-m}.
\]
Hence
\[
    \zeta_\ell(\alpha)\ge2^{-m}\cdot2^{-m}=4^{-m}.
\]
At the endpoints, taking $r=0$ and letting $p\to1$ or $p\to0$, respectively, also gives $\zeta_\ell(0)=\zeta_\ell(1)=1$.
\end{proof}

\subsection{Proofs of Lemmas~\ref{lem:smallConstant},~\ref{lem:simplified upper bound s=o(n)}, and~\ref{lem:exponential_bound}}

\begin{proof}[Proof of Lemma~\ref{lem:smallConstant}]
Recall the formula for $\lambda_\ell$ in~\eqref{eq:constants}. 
Let $t=\binom m\ell^{-1/\ell}$. The definition~\eqref{eq:sigma} gives, on $[0,t]^m$,
\[
    \sigma_\ell(\mathbf w)\le\binom m\ell t^\ell=1.
\]
Since $e^{-\sigma_\ell(\mathbf w)}$ is strictly positive,
\[
    c_\ell
    >\int_{[0,t]^m}e^{-\sigma_\ell(\mathbf w)}\,\d\mathbf w
    \ge t^m e^{-1}.
\]
Consequently,
\begin{align*}
\lambda_\ell&{}>e^{(3\ell-1)/\ell}\binom m\ell^{m/\ell}t^m e^{-1}\\
&{}=e^{(2\ell-1)/\ell}>1.
\end{align*}
Therefore,
\[
    (2\lambda_\ell)^{-\ell/(\ell-1)}
    <2^{-\ell/(\ell-1)}<2^{-1}=\frac12.
\]
\end{proof}

\begin{proof}[Proof of Lemma~\ref{lem:simplified upper bound s=o(n)}]\label{sec:lemmas-3}
Recall the definitions of $\beta$ and $\delta$ in~\eqref{eq:hamiltonian-cutoffs}. 
Lemma~\ref{lem:smallConstant}, proved below directly from the defining integral in~\eqref{eq:constants}, gives $\delta<1/2$. 
Thus Theorem~\ref{thm:upper bound s=o(n)} applies to every integer $3\le s\le\delta n$, and bounds both probabilities by
\[
B_s=\lambda_\ell^s\left(\frac sn\right)^{\beta s}.
\]
Since $\ell\ge2$, we have $3\beta\ge m/\ell$. Choose $N_\ell\ge256$ so that, for every $n\ge N_\ell$,
\begin{equation}\label{eq:hamiltonian-large-n-conditions}
\lambda_\ell n^{-3\beta/4}\le1,
\qquad
2^{-n^{1/4}}\le n^{-m/\ell}.
\end{equation}
Such a threshold depends only on $\ell$: the first left-hand side in~\eqref{eq:hamiltonian-large-n-conditions} tends to zero, and $n^{1/4}/\ln n\to\infty$ gives its second inequality. We now consider three exhaustive ranges of integer $s$. If $s=3$, then
\begin{align*}
B_3&{}=\lambda_\ell^3 3^{3\beta}n^{-3\beta}\\
&{}\le(3\lambda_\ell)^3n^{-m/\ell}.
\end{align*}
If $4\le s\le n^{1/4}$, then $s/n\le n^{-3/4}$ and~\eqref{eq:hamiltonian-large-n-conditions} makes the base in the following power at most one. Hence,
\[
\begin{aligned}
B_s
&\le\bigl(\lambda_\ell n^{-3\beta/4}\bigr)^s
 \le\bigl(\lambda_\ell n^{-3\beta/4}\bigr)^4\\
&=\lambda_\ell^4n^{-3\beta}
 \le\lambda_\ell^4n^{-m/\ell}.
\end{aligned}
\]
Finally, if $n^{1/4}<s\le\delta n$, then $\lambda_\ell(s/n)^\beta\le\lambda_\ell\delta^\beta=1/2$ by~\eqref{eq:hamiltonian-cutoffs}. 
Together with~\eqref{eq:hamiltonian-large-n-conditions}, this gives
\[
B_s\le2^{-s}\le2^{-n^{1/4}}\le n^{-m/\ell}.
\]
Thus $K_\ell=\max\{(3\lambda_\ell)^3,\lambda_\ell^4,1\}$ works simultaneously for every integer in the stated range.
\end{proof}

For $0<\alpha<1$, let $h(\alpha)$ be the binary entropy function using natural logarithms, and set $h(0)=h(1)=0$ by continuity. Thus,
\begin{equation}\label{eq:central-entropy}
    \begin{aligned}
    h(\alpha)&=\alpha\ln\frac1\alpha
        +(1-\alpha)\ln\frac1{1-\alpha},\\
    e^{-h(\alpha)}&=\alpha^\alpha(1-\alpha)^{1-\alpha}.
    \end{aligned}
\end{equation}
The product identity in~\eqref{eq:central-entropy} is understood by continuity at the endpoints; equivalently, use $0\ln0=0$ in the entropy formula. 
Recall the following standard bound.
\begin{lemma}[Method of types]\label{lem:method of types}
For integers $n\ge1$ and $0\le s\le n$, with $\alpha=s/n$,
\[
    \binom ns\le e^{nh(\alpha)}.
\]
\end{lemma}
The proof can be found, for example, in \cite[Example 11.1.3]{cover2006information}.
Now we are ready to prove Lemma~\ref{lem:exponential_bound}.

\begin{proof}[Proof of Lemma~\ref{lem:exponential_bound}]
Fix $\alpha\in[\delta,1-\delta]$. 
Abbreviate the two binomial tails involved in \eqref{eq:zeta} by
\[
    a_r(p)=F_{m,p}(r),
    \qquad
    b_r(p)=F_{m,1-p}(\ell-1-r).
\]
Recall from~\eqref{eq:central-constants} that
\[
    \epsilon_\ell
    =\min\left\{\frac1{\ell\binom m\ell},
        \frac{\delta(1-\delta)^{(1-\delta)/\delta}}2\right\},
\]
which satisfies $0<\epsilon_\ell<\delta/2<1/4$. 
For every integer $r\in[\ell-1]_0$, we divide into three cases according to whether $p$ lies in $(0,\epsilon_\ell)$, $[\epsilon_\ell,1-\epsilon_\ell]$, or $(1-\epsilon_\ell,1)$.

We begin with $p\in[\epsilon_\ell,1-\epsilon_\ell]$. Consider
\[
    a_r(p)+b_r(p)
    =1-\sum_{i=r+1}^{r+\ell-1}\binom mi p^i(1-p)^{m-i}.
\]
Because $\ell-1\ge1$, the sum contains its $i=r+1$ term. For $r\in[\ell-1]_0$ and $p\in[\epsilon_\ell,1-\epsilon_\ell]$,
\begin{align}
a_r(p)+b_r(p)&{}\le1-\binom m{r+1}p^{r+1}(1-p)^{m-r-1}\notag\\
&{}\le1-p^{r+1}(1-p)^{m-r-1}{}\le1-\epsilon_\ell^m.
\label{eq:central-middle-gap}
\end{align}
For $x,y\ge0$, weighted AM--GM gives $(x/\alpha)^\alpha(y/(1-\alpha))^{1-\alpha}\le x+y$, or $x^\alpha y^{1-\alpha}\le(x+y)e^{-h(\alpha)}$. 
Applying this to~\eqref{eq:central-middle-gap} yields
\begin{equation}\label{eq:central-middle-product}
    a_r(p)^\alpha b_r(p)^{1-\alpha}
    \le(1-\epsilon_\ell^m)e^{-h(\alpha)}
    \qquad\text{for }p\in[\epsilon_\ell,1-\epsilon_\ell].
\end{equation}

If $p<\epsilon_\ell$, then
\begin{align*}
b_r(p)&{}=\sum_{i=0}^{\ell-1-r}\binom mi(1-p)^i p^{m-i}\\
&{}\le\sum_{i=0}^{\ell-1-r}\binom m\ell p^{m-i}{}\le\ell\binom m\ell p^\ell
    <\epsilon_\ell^\ell\ell\binom m\ell.
\end{align*}
Since $\ell\ge2$ and $0<\epsilon_\ell<1$, we have $\epsilon_\ell^{\ell-1}\le\epsilon_\ell\le1/(\ell\binom m\ell)$. 
Consequently $b_r(p)<\epsilon_\ell$. As $a_r(p)\le1$,
\begin{align}
a_r(p)^\alpha b_r(p)^{1-\alpha}&{}\le\epsilon_\ell^{1-\alpha}{}\le\left(\frac{\delta(1-\delta)^{(1-\delta)/\delta}}2\right)^{1-\alpha}.
\label{eq:central-small-p-product}
\end{align}
To bound the right-hand side of~\eqref{eq:central-small-p-product}, note that $f(\alpha)=(1-\alpha)\alpha^{\alpha/(1-\alpha)}$ is strictly decreasing on $(0,1)$, since $(\ln f)'(\alpha)=\ln\alpha/(1-\alpha)^2<0$. Since $\alpha\le1-\delta$,
\begin{align}
\left(\frac{\delta(1-\delta)^{(1-\delta)/\delta}}2\right)^{1-\alpha}&{}\le\left(\frac{(1-\alpha)\alpha^{\alpha/(1-\alpha)}}2\right)^{1-\alpha}\notag\\
&{}=2^{-(1-\alpha)}e^{-h(\alpha)}{}\le2^{-\delta}e^{-h(\alpha)}.
\label{eq:central-small-p-entropy}
\end{align}

Finally, if $p>1-\epsilon_\ell$, replacing $(p,r,\alpha)$ by $(1-p,\ell-1-r,1-\alpha)$ interchanges $a_r(p)$ and $b_r(p)$, so the preceding argument applies unchanged and gives $a_r(p)^\alpha b_r(p)^{1-\alpha}\le2^{-\delta}e^{-h(\alpha)}$.
Recall $\rho_\ell$ from~\eqref{eq:central-constants}. The bounds~\eqref{eq:central-middle-product} and~\eqref{eq:central-small-p-product}--\eqref{eq:central-small-p-entropy}, together with this symmetry argument, are uniform in $p$, the integer $r$, and $\alpha\in[\delta,1-\delta]$. 
Taking the supremum in~\eqref{eq:zeta} therefore gives
\begin{equation}\label{eq:central-zeta-gap}
    \zeta_\ell(\alpha)
    =\sup_{\substack{p\in(0,1)\\r\in[\ell-1]_0}}
        a_r(p)^\alpha b_r(p)^{1-\alpha}
    \le\rho_\ell e^{-h(\alpha)}.
\end{equation}
For the second assertion, let $n,s$ satisfy the stated integer restrictions, set $\alpha=s/n$, and fix any $t\ge0$. Lemma~\ref{lem:method of types} and~\eqref{eq:central-zeta-gap} imply
\begin{align*}
\binom ns\bigl(\zeta_\ell(\alpha)+t\bigr)^n
&{}\le e^{nh(\alpha)}
    \bigl(\rho_\ell e^{-h(\alpha)}+t\bigr)^n\\
&{}=\bigl(\rho_\ell+t e^{h(\alpha)}\bigr)^n\\
&{}\le(\rho_\ell+2t)^n.
\end{align*}
The final inequality follows from $h(\alpha)\le\ln2$, the maximum attained at $\alpha=1/2$.
\end{proof}
\section*{Acknowledgments}
The authors are grateful to the MIT PRIMES program for the opportunity to conduct this fascinating research.
AI was used in this paper to help with coding simulations, working through example cases, and proofreading.

\bibliographystyle{alpha}
\bibliography{bib}

\end{document}